\documentclass[sigconf]{acmart}

\renewcommand\footnotetextcopyrightpermission[1]{} 
\usepackage{booktabs} 
\usepackage{amsmath}
\usepackage{amsthm}
\usepackage{amssymb}
\usepackage{graphicx}
\usepackage{algorithm}
\usepackage{algorithmicx}
\usepackage{algpseudocode}
\usepackage{subcaption}
\usepackage{url}
\usepackage{verbatim}

\setcopyright{none}
\acmDOI{}
\acmISBN{}

\theoremstyle{theorem}
\newtheorem{constraint}{Constraint}

\begin{document}
\title{Analysis of Global Fixed-Priority Scheduling for Generalized Sporadic DAG Tasks}

\author{Son Dinh, Christopher Gill, Kunal Agrawal}
\affiliation{%
  \institution{Washington University in Saint Louis, \\
  Department of Computer Science and Engineering}
}
\email{sonndinh, cdgill, kunal@wustl.edu}

\begin{abstract}
We consider global fixed-priority (G-FP) scheduling of parallel tasks, 
in which each task is represented as a directed acyclic graph (DAG). 
We summarize and highlight limitations of the state-of-the-art analyses for G-FP and 
propose a novel technique for bounding interfering workload, which can be applied directly to generalized DAG tasks. 
Our technique works by constructing optimization problems for which 
the optimal solution values serve as safe and tight upper bounds for interfering workloads. 
Using the proposed workload bounding technique, we derive a response-time analysis 
and show that it improves upon state-of-the-art analysis techniques for G-FP scheduling. 
\end{abstract}

\maketitle

\section{Introduction}
\label{sec:introduction}

With the prevalence of multiprocessor platforms and parallel programming languages and runtime systems 
such as OpenMP~\cite{openmp}, Cilk Plus~\cite{frigo1998implementation, cilkplus}, 
and Intel's Threading Building Blocks~\cite{tbb}, the demand for computer programs to be 
able to exploit the parallelism offered by modern hardware is inevitable. In recent 
years, the real-time systems research community has worked to address this trend for 
real-time applications that require parallel execution to satisfy their deadlines, such as real-time hybrid 
simulation of structures~\cite{ferry2014real}, and autonomous vehicles~\cite{kim2013parallel}.

Much effort has been made to develop analysis techniques and schedulability tests for scheduling 
parallel real-time tasks under scheduling algorithms such as Global Earliest Deadline First (G-EDF), 
and Global Deadline Monotonic (G-DM). However, schedulability analysis for parallel tasks is inherently 
more complex than for conventional sequential tasks. 
This is because \emph{intra-task parallelism} is allowed within individual tasks, which enables each individual 
task to execute simultaneously upon multiple processors. The parallelism of each task can also 
vary as it is executing, as it depends on the precedence constraints imposed on the task. 
Consequently, this raises questions of how to account for inter-task interference caused by other 
tasks on a task and intra-task interference among threads of the task itself.

In this paper, we consider task systems that consist of parallel tasks scheduled under Global Fixed-Priority 
(G-FP), in which each task is represented by a Directed Acyclic Graph (DAG). 
Our analysis is based on the concepts of \emph{critical interference} and 
\emph{critical chain}~\cite{chwa2013global, melani2015response, chwa2017global}, 
which allow the analysis to focus on a special 
chain of sequential segments of each task, and hence enable us to use techniques similar to the ones developed for 
sequential tasks~\cite{baker2003multiprocessor, bertogna2005improved, bertogna2007response, 
bertogna2009schedulability}.

The contributions of this paper are as follows:
\begin{itemize}
\item We summarize the state-of-the-art analyses for G-FP and highlight their limitations, specifically 
for the calculation of interference of carry-in jobs and carry-out jobs. 
\item We propose a new technique for computing upper-bounds on carry-out workloads, 
by transforming the problem into an optimization problem that can be solved by modern optimization solvers. 
\item We present a response-time analysis, using the workload bound computed with the new technique. 
Experimental results for randomly generated DAG tasks confirm that our technique dominates existing 
analyses for G-FP. 
\end{itemize}

The rest of this paper is organized as follows. In Sections~\ref{sec:related} and~\ref{sec:model} 
we discuss related work and present the task model we consider in this paper. 
Section~\ref{sec:background} reviews the concepts of \emph{critical interference} and 
\emph{critical chain} and discusses a general framework to bound response-time. 
Section~\ref{sec:sota} summarizes the most recent analyses of G-FP, and also highlights limitations of 
those analyses. In Section~\ref{sec:carryout} we propose a new technique to bound carry-out workload. 
A response-time analysis and a discussion of the complexity of our method are given in 
Section~\ref{sec:rta_schedtest}. Section~\ref{sec:evaluation} presents the evaluation of our method for 
randomly generated DAG tasks. We conclude our work in Section~\ref{sec:conclusion}.

\section{Related Work}
\label{sec:related}

For the sequential task model, Bertogna et al.~\cite{bertogna2007response} proposed a 
response-time analysis that works for G-EDF and G-FP. They bound the interference of 
a task in a problem window by the worst-case workload it can generate 
in that window. The worst-case workload is then bounded by considering a 
worst-case release pattern of the interfering task. This technique was later extended 
by others to analyze parallel tasks, as is done in this work. Bertogna et al.~\cite{bertogna2009schedulability} 
proposed a sufficient slack-based schedulability test for G-EDF and G-FP in which 
the slack values for the tasks are used in an iterative algorithm to improve the schedulability gradually. 
Later, Guan et al.~\cite{guan2009new} proposed a new response-time analysis for both 
constrained-deadline and arbitrary-deadline tasks.

Initially, simple parallel real-time task models were studied, such 
as the fork-join task model and the synchronous task model. Lakshmanan et al.~\cite{lakshmanan2010scheduling} 
presented a transformation algorithm to schedule fork-join tasks where all parallel segments of each task 
must have the same number of threads, which must be less than the number of processors. 
They also proved a resource augmentation bound of 3.42 for their algorithm. Saifullah et al.~\cite{saifullah2013multi} 
improved on that work by removing the restriction on the number of threads in parallel segments. 
They proposed a task decomposition algorithm and proved resource augmentation bounds for the 
algorithm under G-EDF and Partitioned Deadline Monotonic (P-DM) scheduling. 
Axer et al.~\cite{axer2013response} presented a response-time analysis for fork-join tasks under 
Partitioned Fixed-Priority (P-FP) scheduling. Chwa et al.~\cite{chwa2013global} developed an 
analysis for synchronous parallel tasks scheduled under G-EDF. They introduced the concept of 
critical interference and presented a sufficient test for G-EDF. Maia et al.~\cite{maia2014response} 
reused the concept of critical interference to introduce a response-time analysis for synchronous tasks 
scheduled under G-FP. 
A general parallel task model was presented by Baruah et al.~\cite{baruah2012generalized} in which 
each task is modeled as a Directed Acyclic Graph (DAG) and can have an arbitrary deadline. They presented 
a polynomial test and a pseudo-polynomial test for a DAG task scheduled with EDF and proved their 
speedup bounds. However, they only considered a single DAG task. Bonifaci et al.~\cite{bonifaci2013feasibility} 
later developed feasibility tests for task systems with multiple DAG tasks, scheduled under G-EDF and 
G-DM. 

Melani et al.~\cite{melani2015response} proposed a response-time analysis for conditional 
DAG tasks where each DAG can have conditional vertices. Their analysis 
utilizes the concepts of critical interference and critical chain, and works for both G-EDF and G-FP. However, 
the bounds for carry-in and carry-out workloads are likely to be overestimated since they ignore the 
internal structures of the tasks. Chwa et al.~\cite{chwa2017global} extended their work in~\cite{chwa2013global} for 
DAG tasks scheduled under G-EDF. They proposed a sufficient, workload-based schedulability test 
and improved it by exploiting slack values of the tasks. 
Fonseca et al.~\cite{fonseca2017improved} proposed a response-time 
analysis for sporadic DAG tasks scheduled under G-FP that improves upon the response-time analysis 
in~\cite{melani2015response}. They improve the upper bounds for interference by taking the DAGs of the 
tasks into consideration. In particular, by explicitly considering the DAGs the workloads generated by 
the carry-in and carry-out jobs can be reduced compared to the ones in~\cite{melani2015response}, 
and hence schedulability can be improved. The carry-in workload is bounded by considering a schedule 
for the carry-in job with unrestricted processors, in which subtasks execute as soon as they are ready and 
for their full WCETs. The carry-out workload is bounded for a less general type of DAG tasks, called 
nested fork-join DAGs. 
We discuss the state-of-the-art analyses for G-FP and differentiate our work in detail in Section~\ref{sec:sota}.

\section{System Model}
\label{sec:model}

\begin{figure}[ht]
\centering
\includegraphics[width=\linewidth]{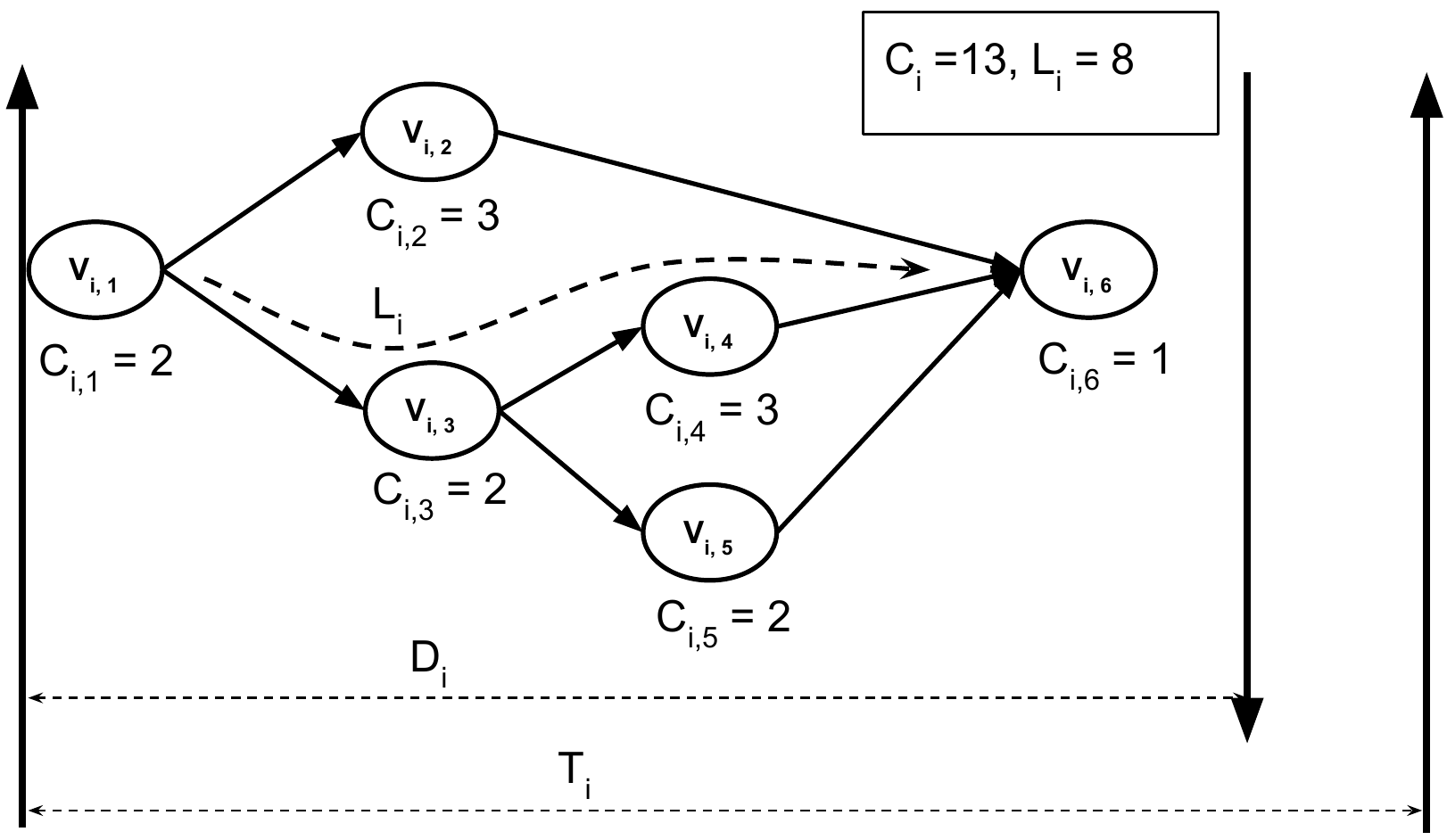}
\caption{An example DAG task.}
\label{fig:example_task}
\end{figure}

We consider a set $\tau$ of $n$ real-time parallel tasks, $\tau = \{\tau_1, \tau_2, ..., \tau_n\}$, 
scheduled preemptively by a global fixed-priority scheduling algorithm upon $m$ identical processors. 
Each task $\tau_i$ is a recurrent, sporadic 
process which may release an infinite sequence of jobs and is modeled by $\tau_i = \{G_i, D_i, T_i\}$, 
where $D_i$ denotes its relative deadline and $T_i$ denotes the minimum inter-arrival time of 
two consecutive jobs of $\tau_i$. We assume that all tasks have constrained deadlines, i.e., 
$D_i\leq T_i, \forall i\in [1, n]$. 
Each task $\tau_i$ is represented as a Directed Acyclic Graph (DAG) $G_i=(V_i, E_i)$, 
where $V_i=\{v_{i, 1}, v_{i, 2}, ..., v_{i, n_i}\}$ is the set of vertices of the DAG $G_i$ and 
$E_i\subseteq (V_i\times V_i)$ is the set of directed edges of $G_i$. 
In this paper, we also use \emph{subtasks} and \emph{nodes} to refer to the vertices of the tasks. 
Each subtask $v_{i, a}$ of $G_i$ represents a section of instructions that can only be run sequentially. 
A subtask $v_{i, a}$ is called a \emph{predecessor} of $v_{i, b}$ if there exists an edge from $v_{i, a}$ to 
$v_{i, b}$ in $G_i$, i.e., $(v_{i, a}, v_{i, b})\in G_i$. Subtask $v_{i, b}$ is then called a \emph{successor} of 
$v_{i, a}$. Each edge $(v_{i, a}, v_{i, b})$ represents a precedence constraint between the two subtasks. 
A subtask is \emph{ready} if all of its predecessors have finished. 
Whenever a task releases a job, all of its subtasks are released and have the 
same deadline as the job's deadline. We use $J_i$ to denote an arbitrary job of $\tau_i$ 
which has release time $r_i$ and absolute deadline $d_i$.

Each subtask $v_{i, a}$ has a worst-case execution time (WCET), denoted by $C_{i, a}$. 
The sum of WCETs of all subtasks of $\tau_i$ is the worst-case execution time of 
the whole task, and is denoted by $C_i = \sum_{v_{i, a}\in V_i} C_{i, a}$. The WCET of 
a task is also called its \emph{work}. A sequence of subtasks $(v_{i, u_1}, v_{i, u_2}, ..., v_{i, u_t})$ 
of $\tau_i$, in which $(v_{i, u_j}, v_{i, u_{j+1}})\in E_i, \forall 1\leq j\leq t-1$, 
is called a \emph{chain} of $\tau_i$ and is denoted by $\lambda_i$. The length of a chain $\lambda_i$ 
is the sum of the WCETs of subtasks in $\lambda_i$ and is denoted by $len(\lambda_i)$, i.e., 
$len(\lambda_i) = \sum_{v_{i, u_j}\in\lambda_i} C_{i, u_j}$. A chain of $\tau_i$ which has 
the longest length is a \emph{critical path} of the task. 
The length of a critical path of a DAG is called its \emph{critical path length} or 
\emph{span}, and is denoted by $L_i$. 
Figure~\ref{fig:example_task} illustrates 
an example DAG task $\tau_i$ with 6 subtasks, whose work and span are 
$C_i = 13$ and $L_i = 8$, respectively. In this paper, we consider tasks that 
are scheduled using a preemptive, global fixed-priority algorithm where each task is assigned a 
fixed task-level priority. All subtasks of a task have the same priority as the task. 
Without loss of generality, we assume that tasks have distinct priorities, 
and $\tau_i$ has higher priority than $\tau_k$ if $i < k$.

\section{Background}
\label{sec:background}

In this section we discuss the concept of critical interference that our work is based on, 
and present a general framework to bound response-times of DAG tasks scheduled under G-FP. 
In the next section, we summarize the state-of-the-art analyses for G-FP and 
give an overview of our method. 

\subsection{Critical Chain and Critical Interference}
\label{subsec:critical_interference}
The notions of \emph{critical chain} and \emph{critical interference} were introduced by 
Chwa et al.~\cite{chwa2017global, chwa2013global} for analyzing parallel tasks scheduled with G-EDF. 
Unlike sequential tasks, analysis of 
DAG tasks with internal parallelism is inherently more complicated: 
(i) some subtasks of a task can be interfered with by other subtasks of the same task 
(i.e., \emph{intra-task interference}); (ii) subtasks of a task can be interfered with by subtasks of 
higher-priority tasks (i.e., \emph{inter-task interference}); and (iii) the parallelism of a DAG task 
may vary during execution, subject to the precedence constraints imposed by its graph. 
The critical chain and critical interference concepts alleviate the complexity of the analysis 
by focusing on a special chain of subtasks of a task which accounts for its response time, 
thus bringing the problem closer to a more familiar 
analysis technique for sequential tasks. Although they were originally proposed for analysis of 
G-EDF~\cite{chwa2017global, chwa2013global}, these concepts are also useful for analyzing G-FP. 
We therefore use them in our analysis and include a discussion of them in this section.

Consider any job $J_k$ of a task $\tau_k$ and its corresponding schedule. 
A \emph{last-completing subtask} of $J_k$ is a subtask 
that completes last among all subtasks in the schedule of $J_k$. 
A \emph{last-completing predecessor} of a subtask $v_{k, a}$ is a predecessor 
that completes last among all predecessors of $v_{k, a}$ in the schedule of $J_k$. 
Note that a subtask can only be ready after a last-completing predecessor finishes, 
since only then are all the precedence constraints for the subtask satisfied. 
Starting from a last-completing subtask of $J_k$, we can recursively trace back through all 
last-completing predecessors until we reach a subtask with no predecessors. 
If during that process, a subtask has more than one last-completing predecessors, 
we arbitrarily pick one. The chain that is reconstructed by appending those 
last-completing predecessors and the last-completing subtask 
is called a \emph{critical chain} of job $J_k$. We call the subtasks that belong to a critical 
chain \emph{critical subtasks}. 

\begin{figure}[h]
\center
\includegraphics[width=\linewidth]{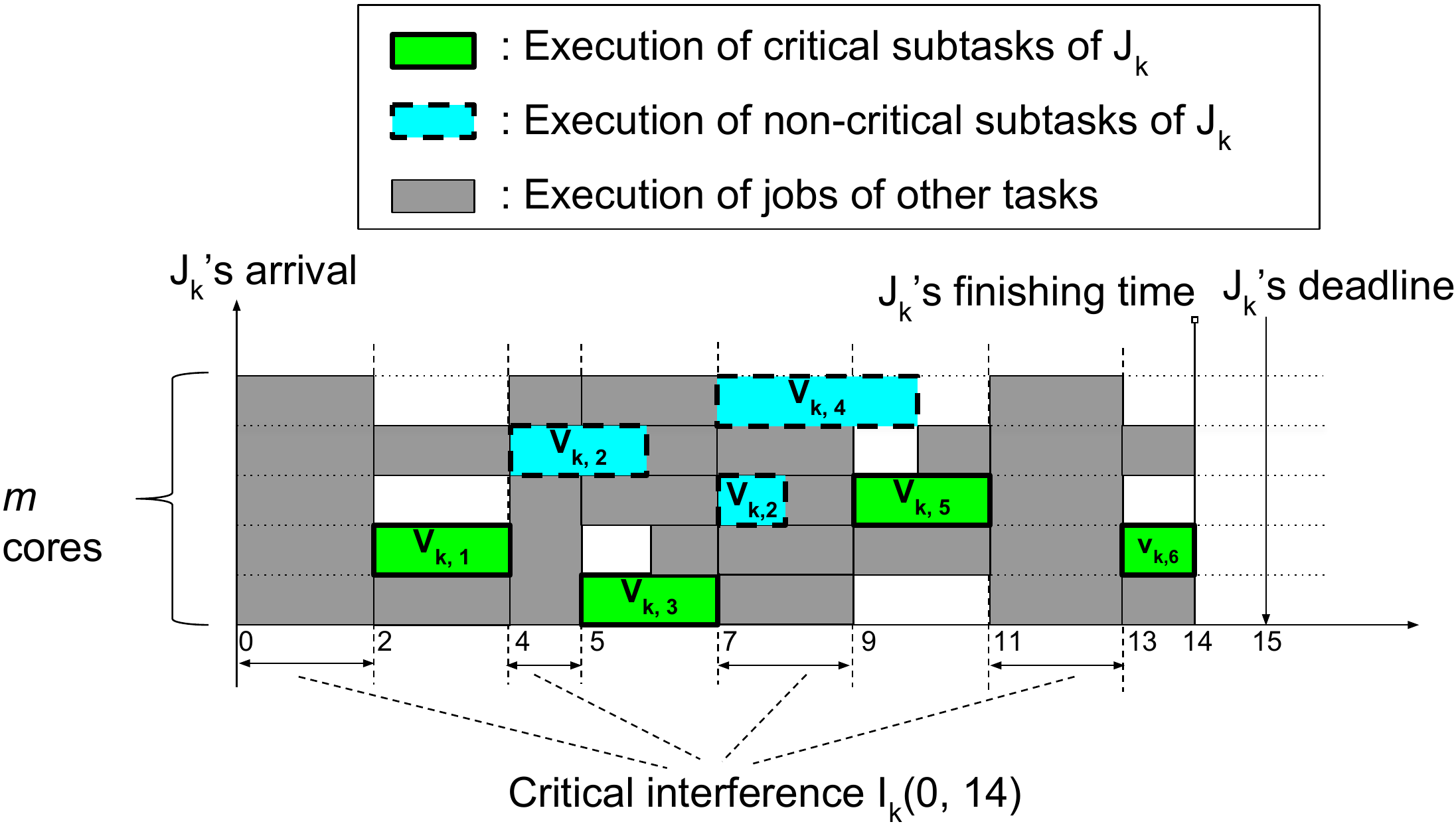}
\caption{Critical chain and critical interference of $J_k$.}
\label{fig:critical_chain}
\end{figure}

\begin{example} 
Figure~\ref{fig:critical_chain} presents an example of a critical chain of a job $J_k$ of 
task $\tau_k$, which has the same DAG as shown in Figure~\ref{fig:example_task}. In Figure~\ref{fig:critical_chain}, 
boxes with bold, solid borders denote the execution of critical subtasks of $J_k$; boxes with bold, 
dashed borders denote the execution of the other subtasks of $J_k$. The other boxes are 
for jobs of other tasks. Subtask $v_{k, 6}$ is a last-completing subtask. 
A last-completing predecessor of $v_{k, 6}$ is $v_{k, 5}$. 
Similarly, a last-completing predecessor of $v_{k, 5}$ is $v_{k, 3}$, and a last-completing 
predecessor of $v_{k, 3}$ is $v_{k, 1}$. Hence a critical chain of $J_k$ is 
$(v_{k, 1}, v_{k, 3}, v_{k, 5}, v_{k, 6})$. 
\end{example}

The critical chain concept has a few properties 
that make it useful for schedulability analysis of parallel DAG tasks. 
First, the first subtask of any critical chain of a job is 
ready to execute as soon as the job is released, since it does not have any predecessor. 
Second, when the last subtask of a critical chain completes, the corresponding job finishes --- 
this is from the construction of the critical chain. Thus the scheduling window of a critical chain 
of $J_k$ --- i.e., from the release time of its first subtask to the completion time of 
its last subtask --- is also the scheduling window of job $J_k$ --- i.e., from the job's release 
time to its completion time. 
Third, consider a critical chain $\lambda_k$ of $J_k$: 
at any instant during the scheduling window of $J_k$, 
either a critical subtask of $\lambda_k$ is executed \textit{or} a critical subtask of $\lambda_k$ 
is ready but not executed because all $m$ processors 
are busy executing subtasks not belonging to $\lambda_k$, 
including non-critical subtasks of job $J_k$ and subtasks from other tasks 
(see Figure~\ref{fig:critical_chain}). 
Therefore, the response-time of a critical chain of $J_k$ is also the response-time of $J_k$. 
Hence if we can upper-bound the response-time of a critical chain for any job $J_k$ of 
$\tau_k$, that bound also serves as an upper-bound for the response-time of $\tau_k$.

The third property of the critical chain suggests that we can partition the scheduling window of 
a job $J_k$ into two sets of intervals. One includes all intervals during which critical subtasks of 
$J_k$ are executed and the other includes all intervals during which a critical subtask of $J_k$ 
is ready but not executed. The total length of the intervals in the second set is called the 
\emph{critical interference} of $J_k$. We include definitions for critical interference and 
interference caused by an individual task on $\tau_k$ as follows. 
\begin{definition}
\label{defn:critical_interference}
\emph{Critical interference} $I_k(a, b)$ on a job $J_k$ of task $\tau_k$ is the aggregated \textbf{length} 
of all intervals in [a, b) during which a critical subtask of $J_k$ is ready but not executed. 
\end{definition}

\begin{definition}
\label{defn:critical_interfering_indi}
\emph{Critical interference} $I_{i, k}(a, b)$ on a job $J_k$ of task $\tau_k$ due to task $\tau_i$ is the aggregated 
\textbf{processor time} from all intervals in [a, b) during which one or more subtasks of $\tau_i$ are executed and 
a critical subtask of $J_k$ is ready but not executed. 
\end{definition}

In Figure~\ref{fig:critical_chain}, the critical interference $I_k(0, 14)$ of $J_k$ is the sum of the lengths of intervals 
$[0, 2)$, $[4, 5)$, $[7, 9)$, and $[11, 13)$ which is 7. The critical interference $I_{i, k}(0, 14)$ 
caused by a task $\tau_i$ is the total processor time of $\tau_i$ in those four intervals. Note 
that $\tau_i$ may execute simultaneously on multiple processors, and we must sum its processor time 
on all processors. From the definition of critical interference, we have: 
\begin{equation}
\label{eqn:critical_interference}
I_k(a, b) = \frac{1}{m} \sum\limits_{\tau_i\in\tau} I_{i, k}(a, b).
\end{equation}

\subsection{A General Method for Bounding Response-Time}
\label{subsec:general_method}

\begin{figure*}
\centering
\includegraphics[width=0.9\textwidth]{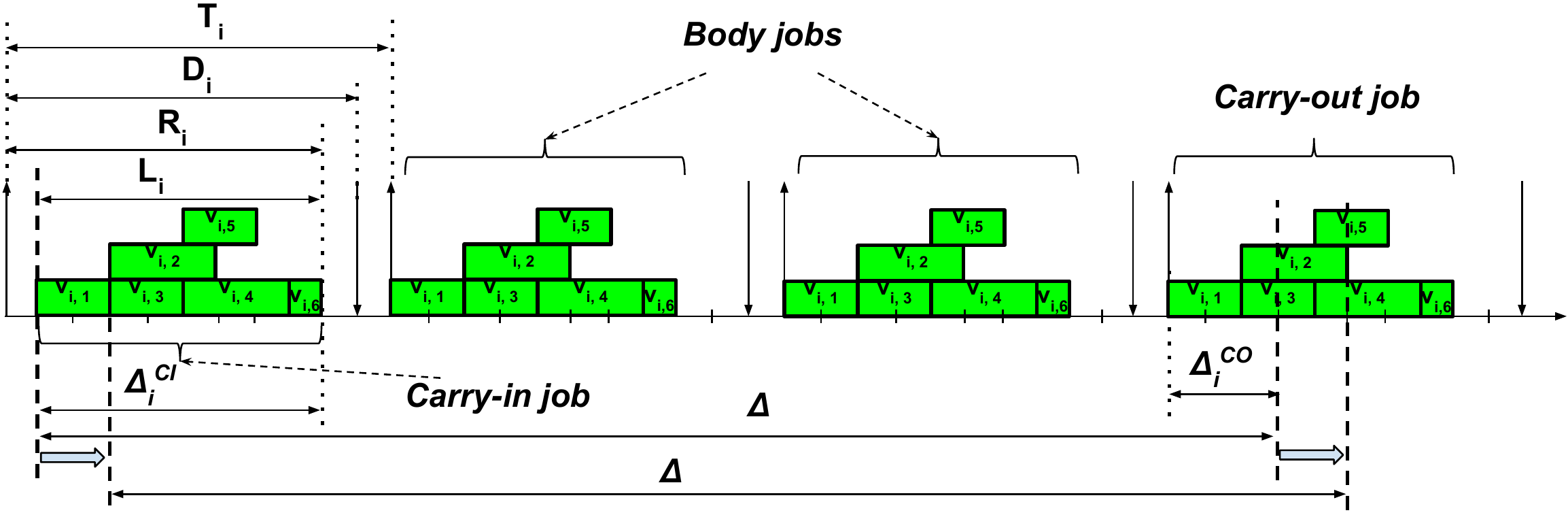}
\caption{Workload generated by an interfering task $\tau_i$ in an interval of length $\Delta$.}
\label{fig:workload}
\end{figure*}

We now discuss a general framework for bounding response-time in G-FP that is used in this work 
and was also employed by the state-of-the-art analyses~\cite{melani2015response, fonseca2017improved}. 
Based on the definitions of critical chain and critical interference, the response-time $R_k$ of $J_k$ is: 
$$R_k = len(\lambda_k) + I_k(r_k, r_k+R_k),$$
where $\lambda_k$ is a critical chain of $J_k$ and $len(\lambda_k)$ is its length 
(see Figure~\ref{fig:critical_chain} for example). 
Applying Equation~\ref{eqn:critical_interference} we have: 
\begin{multline}
\label{eqn:resptime}
R_k = \Big( len(\lambda_k) + \frac{1}{m} I_{k, k}(r_k, r_k+R_k) \Big) + 
\frac{1}{m} \sum\limits_{\tau_i\in hp(\tau_k)} I_{i, k}(r_k, r_k+R_k), 
\end{multline}
where $hp(\tau_k)$ is the set of tasks with higher priorities than $\tau_k$'s. 
Thus if we can bound the 
right-hand side of Equation~\ref{eqn:resptime}, we can bound the response-time of $\tau_k$. 
To do so, we bound the contributions to $J_k$'s response-time 
caused by subtasks of $J_k$ itself and by jobs of higher-priority tasks separately.

\subsubsection{Intra-Task Interference}
The sum $\Big( len(\lambda_k) + \frac{1}{m} I_{k, k}(r_k, r_k+R_k) \Big)$, which includes 
the intra-task interference on the critical chain of $J_k$ caused by non-critical subtasks of $J_k$, 
is bounded by Lemma V.3 in~\cite{melani2015response}. 
We include the bound below. 
\begin{lemma}
\label{lem:intra_interference}
The following inequality holds for any task $\tau_k$ scheduled by any work-conserving algorithm: \\
$$len(\lambda_k) + \frac{1}{m} I_{k, k}(r_k, r_k+R_k)\leq L_k + \frac{1}{m} (C_k - L_k)$$
\end{lemma}

\subsubsection{Inter-Task Interference}
Now we need to bound the inter-task interference on the right-hand side of Equation~\ref{eqn:resptime}. 
Since the interference caused by a task in an interval is at most the workload generated 
by the task during that interval, we can bound $I_{i, k}(a, b),~\forall\tau_i\in hp(\tau_k)$ 
using the bound for the workload generated by $\tau_i$ in the interval $[a, b)$. 
Let $W_i(a, b)$ denote the maximum workload generated by $\tau_i$ in the interval $[a, b)$. 
Let $W_i(\Delta)$ denote the maximum workload generated by $\tau_i$ in any interval of 
length $\Delta$. The following inequality holds for any $\tau_i$: 
\begin{equation}
\label{eqn:workload_relation}
I_{i, k}(r_k, r_k+R_k)\leq W_i(r_k, r_k+R_k)\leq W_i(R_k).
\end{equation}
Let the \emph{problem window} be the interval of interest with length $\Delta$. 
The jobs of $\tau_i$ that may generate workload within the problem window are classified into three types: 
(i) A \emph{carry-in job} is released strictly before the problem window and has a deadline within it, 
(ii) A \emph{carry-out job} is released within the problem window and has its deadline strictly after it, 
and (iii) \emph{body jobs} have both release time and deadline within the problem window. 
Similar to analyses for sequential tasks (e.g., Bertogna et al.~\cite{bertogna2007response}), 
the maximum workload generated by $\tau_i$ in 
the problem window can be attained with a release pattern in which (i) jobs of $\tau_i$ are released 
as quickly as possible, meaning that the gap between any two consecutive releases is exactly 
the period $T_i$, (ii) the carry-in job finishes as late as its worst-case 
finishing time, and (iii) the body jobs and the carry-out job start executing as soon as 
they are released. 
Figure~\ref{fig:workload} shows an example of such a job-release pattern of an interfering 
task $\tau_i$ with the DAG structure shown in Figure~\ref{fig:example_task}. 

However, unlike sequential tasks, analysis for parallel DAG tasks is more challenging in 
two aspects. First, it is not obvious which schedule for the subtasks of the carry-in (carry-out) job 
would generate maximum carry-in (carry-out) workload. This is because the parallelism 
of a DAG task can vary depending on its internal graph structure. Second, for the same reason, aligning 
the problem window's start time with the start time of the carry-in job of $\tau_i$ may not correspond to the 
maximum workload generated by $\tau_i$. For instance, in Figure~\ref{fig:workload} if we shift 
the problem window to the right 2 time units, the carry-in 
job's workload loses 2 time units but the carry-out job's workload gains 5 time units. 
The total workload thus increases 3 time units. Therefore in order to compute the 
maximum workload generated by $\tau_i$ we must slide the problem window to find 
a position that corresponds to the maximum sum of the carry-in workload and carry-out workload. 
We discuss an existing method for computing carry-in workload in Section~\ref{sec:sota} and 
our technique for computing carry-out workload in Section~\ref{sec:carryout}. 
In Section~\ref{sec:rta_schedtest}, we combine those two bounds in a 
response-time analysis and explain how we slide problem windows to compute maximum workloads. 

We note that the maximum workload generated by each body job does not depend 
on the schedule of its subtasks and is simply its total work. 
Furthermore, regardless of the position of the problem window, the workload contributed by 
the body jobs, denoted by $W_i^{BO}(\Delta)$, is bounded as follows. 
\begin{lemma}
The workload generated by the body jobs of task $\tau_i$ in a problem window with length $\Delta$ 
is upper-bounded by 
$$W_i^{BO}(\Delta) = \max \Big\{ \big(\Bigl\lfloor \frac{\Delta-L_i+R_i}{T_i} \Bigr\rfloor - 1\big) C_i,  0\Big\}. $$
\end{lemma}
\begin{proof}
Consider the case where the start of the problem window is aligned with the starting time of 
the carry-in job, as shown in Figure~\ref{fig:workload}. The number of body jobs is at most 
$\max\big\{\Bigl\lfloor \frac{\Delta-L_i+R_i}{T_i} \Bigr\rfloor - 1, 0\big\}$. 
Thus for this case the workload 
of the body jobs is at most $\max \big\{ \big(\Bigl\lfloor \frac{\Delta-L_i+R_i}{T_i} \Bigr\rfloor - 1\big) C_i, 0\big\}$. 

Shifting the problem window to the left or right can change the workload contributed by the 
carry-in and carry-out jobs but does not increase the maximum number of body jobs or their 
workload. The bound thus follows. 
\end{proof}

Let the \emph{carry-in window} and \emph{carry-out window} be the intervals within the problem 
window during which the carry-in job and the carry-out job are executed, respectively. 
Intuitively, the carry-in window spans from the start of the problem window to the completion time
of the carry-in job; the carry-out window spans from the starting time of the carry-out job to the 
end of the problem window. We denote the lengths of the carry-in window and carry-out window 
for task $\tau_i$ by $\Delta_i^{CI}$ and $\Delta_i^{CO}$ respectively. 
The sum of $\Delta_i^{CI}$ and $\Delta_i^{CO}$ is: 
\begin{equation}
\label{eqn:ci_co_length} 
\Delta_i^{CI} + \Delta_i^{CO} = L_i + (\Delta-L_i+R_i)\mod T_i
\end{equation}
Let $W_i^{CI}(\Delta_i^{CI})$ be the maximum carry-in workload of $\tau_i$ for a carry-in window 
of length $\Delta_i^{CI}$. Similarly, let $W_i^{CO}(\Delta_i^{CO})$ be the maximum carry-out workload 
of $\tau_i$ for a carry-out window of length $\Delta_i^{CO}$. 
The maximum workload generated by $\tau_i$ in any problem window of length $\Delta$ can be 
computed by taking the maximum over all $\Delta_i^{CI}$ and $\Delta_i^{CO}$ that satisfy 
Equation~\ref{eqn:ci_co_length}: 
\begin{equation}
\label{eqn:max_workload}
W_i(\Delta) = W_i^{BO}(\Delta) + \max\limits_{\Delta_i^{CI}, \Delta_i^{CO} \text{satisfy Eq.~\ref{eqn:ci_co_length}}} 
\Big\{ W_i^{CI}(\Delta_i^{CI}) + W_i^{CO}(\Delta_i^{CO}) \Big\}.
\end{equation}
Therefore if we can bound $W_i^{CI}(\Delta_i^{CI})$ and $W_i^{CO}(\Delta_i^{CO})$, we can 
bound the inter-task interference of $\tau_i$ on $\tau_k$ and thus the response-time of $\tau_k$.

\section{The State-of-the-Art Analysis for G-FP}
\label{sec:sota}

\begin{figure}[h!]
\centering
\includegraphics[width=\linewidth]{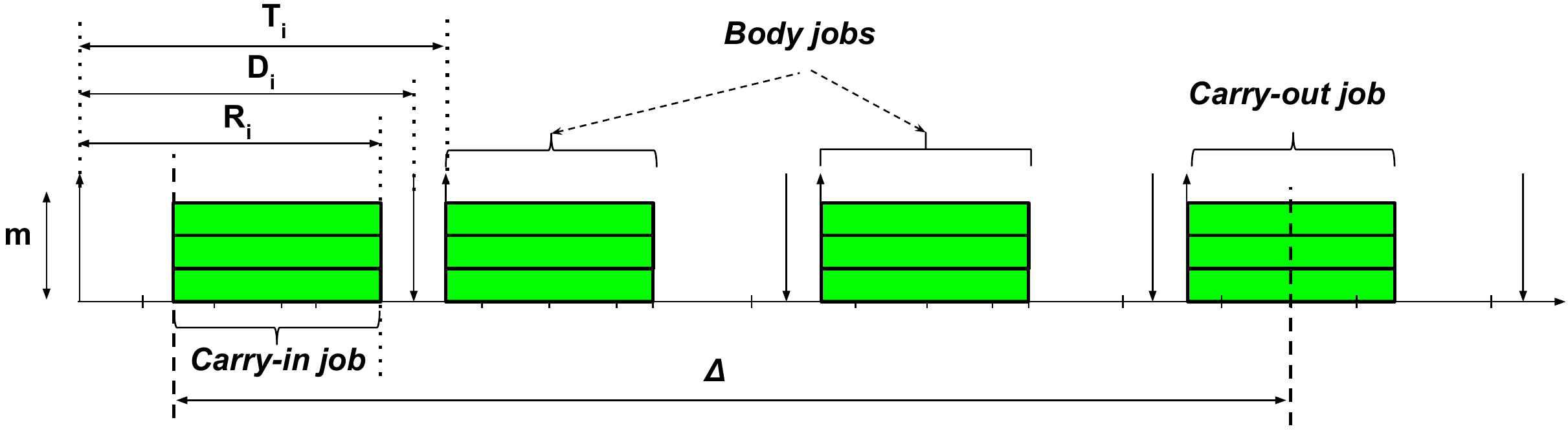}
\caption{Workload generated by an interfering task $\tau_i$ in Melani et al.~\cite{melani2015response}.}
\label{fig:melani}
\end{figure}

Melani et al.~\cite{melani2015response} proposed a response-time analysis for G-FP scheduling 
of conditional DAG tasks that may contain conditional vertices, for modeling conditional constructs 
such as \texttt{if-then-else} statements. 
They bounded the interfering workload by assuming that jobs of the interfering task execute perfectly 
in parallel on all $m$ processors. Their bound for the interfering workload is computed as follows. 
$$W_i(\Delta) = \Big\lfloor\frac{\Delta+R_i-C_i/m}{T_i} \Big\rfloor C_i + 
\min\big\{ C_i, m((\Delta+R_i-C_i/m)\mod T_i) \big\}.$$
Figure~\ref{fig:melani} illustrates the workload computation for an interfering task $\tau_i$ given 
in~\cite{melani2015response}. As shown in this figure, both carry-in and carry-out jobs are assumed to 
execute with perfect parallelism upon $m$ processors. Thus their workload contributions in the considered 
window are maximized. This assumption simplifies the workload computation as it ignores the 
internal DAG structures of the interfering tasks. However, assuming that DAG tasks have such 
abundant parallelism is likely unrealistic and thus makes the analysis pessimistic. 

Fonseca et al.~\cite{fonseca2017improved} later considered a task model similar to the one in this paper 
and proposed a method to improve the bounds for carry-in and carry-out 
workloads by explicitly considering the DAGs. The carry-in workload was bounded using a 
\emph{hypothetical schedule} for the carry-in job, in which the carry-in job can use as many processors 
as it needs to fully exploit its parallelism. They proved that the carry-in workload of the hypothetical 
schedule is maximized when: (i) the hypothetical schedule's completion time is aligned with the worst-case 
completion time of the interfering task, (ii) every subtask in the hypothetical schedule starts executing 
as soon as all of its predecessors finish, and (iii) every subtask in the hypothetical 
schedule executes for its full WCET. Figure~\ref{fig:workload} shows the hypothetical schedule 
of the carry-in job for the task in Figure~\ref{fig:example_task}. In this paper, we adopt their method for 
computing carry-in workload. In particular, the carry-in workload of task $\tau_i$ with a carry-in window of 
length $\Delta_i^{CI}$, i.e., from the start of the problem window to the completion time of the carry-in job 
(see Figure~\ref{fig:workload}), is computed as follows.
\begin{equation}
\label{eqn:carryin}
W_i^{CI}(\Delta_i^{CI}) = \sum_{v_{i, k}\in V_i}\max\big\{ C_{i, k} - \max(L_i - S_{i, k} - \Delta_i^{CI}, 0), 0 \big\}.
\end{equation}
In Equation~\ref{eqn:carryin}, $S_{i, k}$ is the start time of subtask $v_{i, k}$ in the hypothetical schedule for 
the carry-in job described above. It can be computed by taking a longest path among all paths from source 
subtasks to $v_{i, k}$ and adding up the WCETs of the subtasks along that path excluding $v_{i, k}$ itself. 

\begin{figure}[t!]
\centering
\begin{subfigure}{0.7\linewidth}
  \centering
  \includegraphics[width=\linewidth]{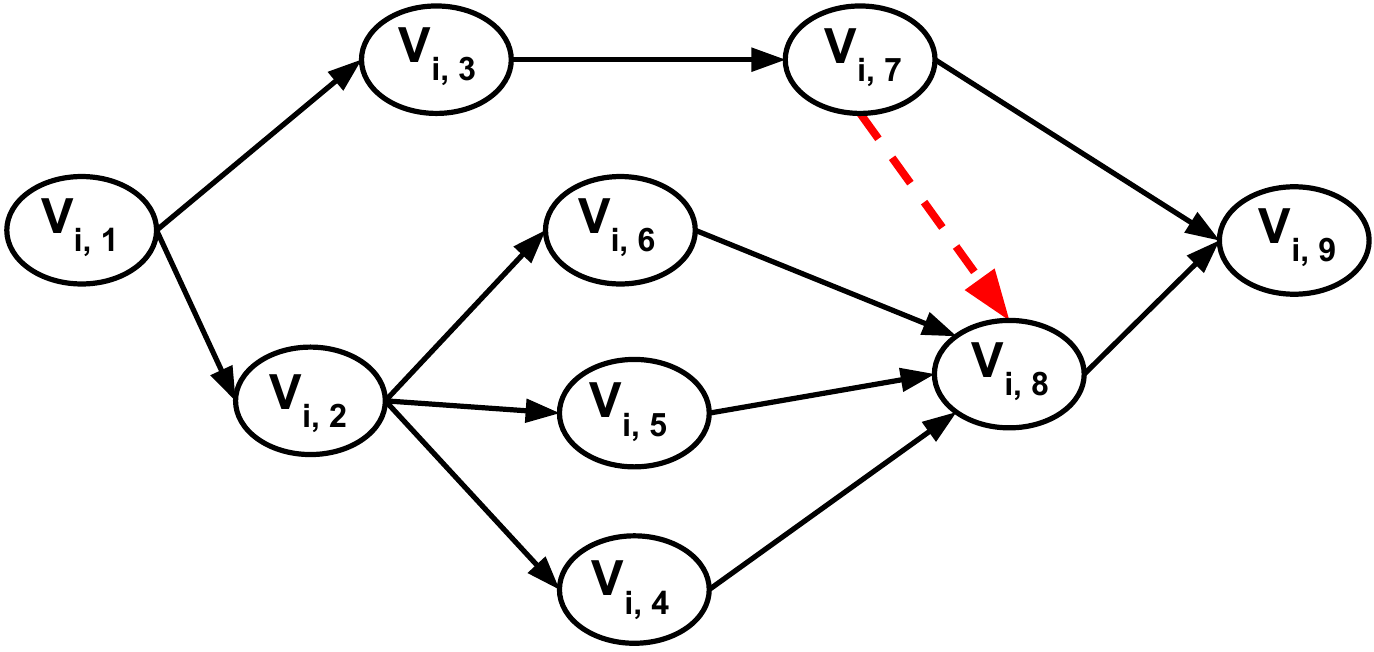}
  \caption{A non-nested-fork-join DAG task.}
  \label{fig:nfjdag_conflict}
\end{subfigure}
\hfill
\begin{subfigure}{.7\linewidth}
  \centering
  \includegraphics[width=\linewidth]{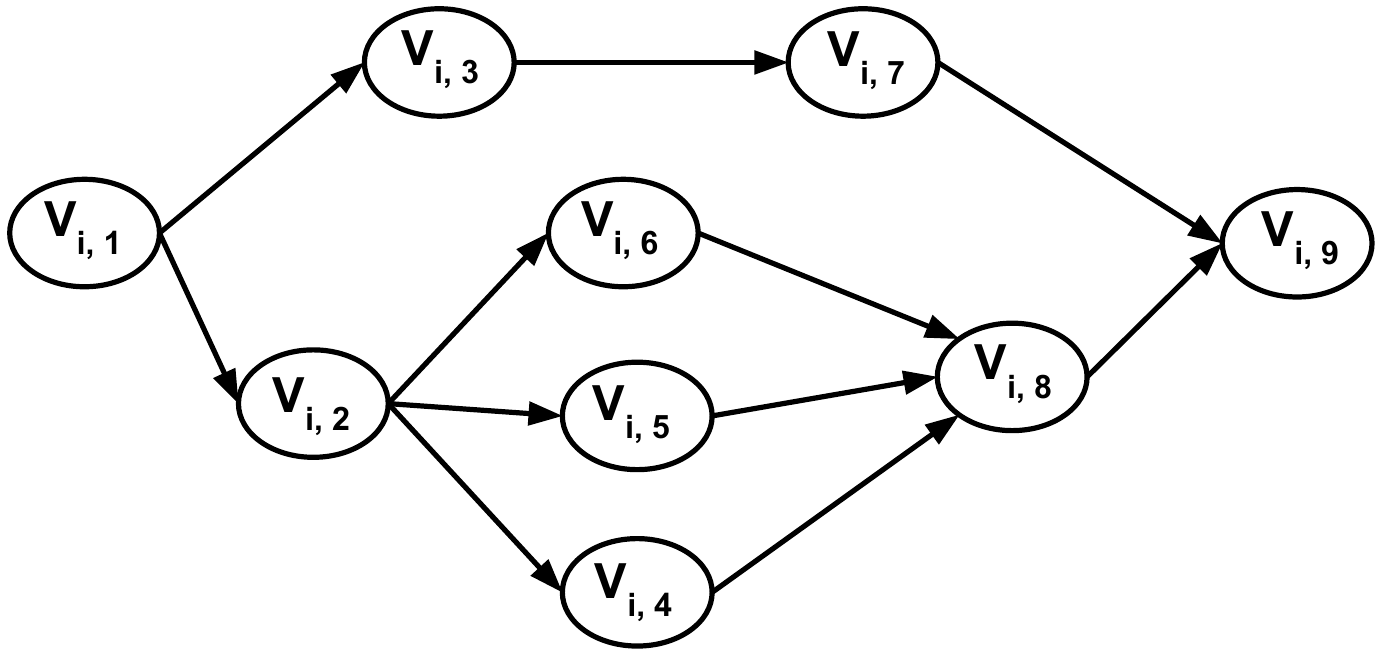}
  \caption{A nested fork-join DAG task.}
  \label{fig:nfjdag}
\end{subfigure}
\caption{Example for a general DAG task and a nested fork-join DAG task.}
\label{fig:nfjdagtask}
\end{figure}

For the carry-out workload,~\cite{fonseca2017improved} considered a subset of generalized DAG tasks, 
namely nested fork-join DAG (NFJ-DAG) tasks. A NFJ-DAG is constructed recursively from smaller NFJ-DAGs 
using two operations: \emph{series composition} and \emph{parallel composition}. Figure~\ref{fig:nfjdag} 
shows an example NFJ-DAG task. Figure~\ref{fig:nfjdag_conflict} shows a similar DAG with one more 
edge $(v_{i, 7}, v_{i, 8})$. The DAG in Figure~\ref{fig:nfjdag_conflict} is not a NFJ-DAG due to a single 
cross edge $(v_{i, 7}, v_{i, 8})$. To deal with a non NFJ-DAG,~\cite{fonseca2017improved} first transforms 
the original DAG to a NFJ-DAG by removing the conflicting edges, such as $(v_{i, 7}, v_{i, 8})$ in 
Figure~\ref{fig:nfjdagtask}. Then they compute the upper-bound for the carry-out workload using the 
obtained NFJ-DAG. The computed bound is proved to be an upper-bound for the carry-out workload. 
We note that the transformation removes some precedence constraints from the original DAG, and thus 
the resulting NFJ-DAG may have higher parallelism than the original DAG. Hence, computing the carry-out 
workload of a generalized DAG task via its transformed NFG-DAG may be pessimistic, especially for 
a complex DAG, as the transformation may remove many edges from the original DAG.

In this paper, we propose a new technique to directly compute an upper-bound for the carry-out workload of 
generalized DAG task. The high level idea is to frame the problem of finding the bound as 
an optimization problem, which can be solved effectively by solvers such as 
the CPLEX~\cite{cplex}, Gurobi~\cite{gurobi}, or SCIP~\cite{scip}. 
The solution of the optimization problem then serves as a safe and tight upper-bound for the 
carry-out workload. In the next section we present our method in detail.

\section{Bound for Carry-Out Workload}
\label{sec:carryout}

\begin{figure}[h]
\centering
\begin{subfigure}{0.5\linewidth}
  \centering
  \includegraphics[width=\linewidth]{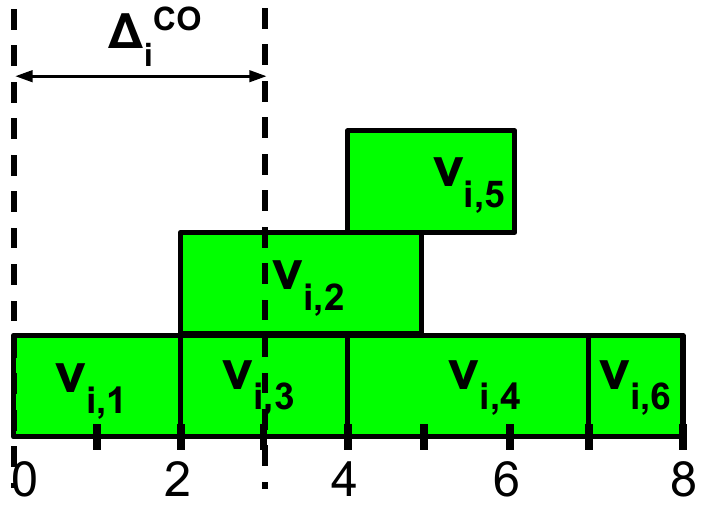}
  \caption{Carry-out workload when all subtasks execute for WCETs.}
  \label{fig:carryout_wcet}
\end{subfigure}
\hfill
\begin{subfigure}{0.4\linewidth}
  \centering
  \includegraphics[width=\linewidth]{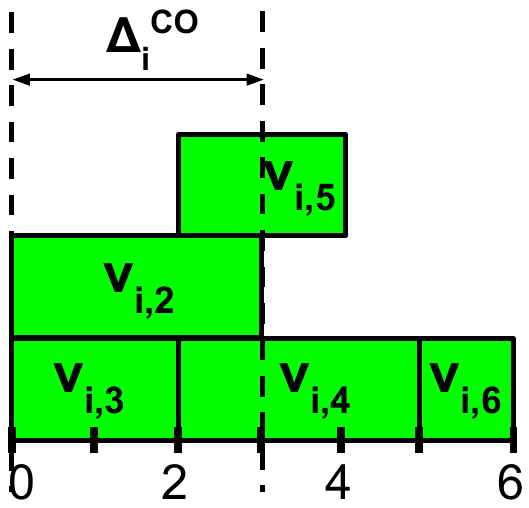}
  \caption{Carry-out workload when subtasks may execute less than WCETs.}
  \label{fig:carryout_nonwcet}
\end{subfigure}

\caption{An illustration of generating the maximum carry-out workload.}
\label{fig:carryout}
\end{figure}

In this section we propose a method to bound the carry-out workload that can be 
generated by a job of task $\tau_i$ by constructing an \textbf{integer linear program} (ILP) 
for which the optimal solution value is an upper-bound of the carry-out workload. 

Consider a carry-out job of task $\tau_i$, which is scheduled with an unrestricted number of processors, 
meaning that it can use as many processors as it requires to fully exploit its parallelism. 
Each subtask of the carry-out job executes as soon as it is ready, i.e., immediately after all of its predecessors 
have finished. 
We label such a schedule for the carry-out job $\mathcal{SCHE}^{CO}(\tau_i)$. 
We prove in the following lemma that the workload generated by $\mathcal{SCHE}^{CO}(\tau_i)$ is 
an upper-bound for the carry-out workload.
\begin{lemma}
\label{lem:asap}
For specific values of the execution times for the subtasks of $\tau_i$, workload generated by 
$\mathcal{SCHE}^{CO}(\tau_i)$ in a carry-out window of length $\Delta_i^{CO}$ is an upper-bound 
for the carry-out workload generated by $\tau_i$ with the given subtasks's execution times. 
\end{lemma}
\begin{proof}
We prove by contradiction. Consider a schedule $\mathcal{SCHE}^*$ for the carry-out job in which 
subtasks execute for the same lengths as in $\mathcal{SCHE}^{CO}(\tau_i)$. Suppose 
subtask $v_{i, k}$ is the first subtask in time order that produces more workload in $\mathcal{SCHE}^*$ 
than it does in $\mathcal{SCHE}^{CO}(\tau_i)$. This means $v_{i, k}$ must have started executing earlier in 
$\mathcal{SCHE}^*$ than it have in $\mathcal{SCHE}^{CO}(\tau_i)$. Hence, $v_{i, k}$ must have 
started its execution before all of its predecessors have finished in $\mathcal{SCHE}^*$. This is 
impossible and the lemma follows. 
\end{proof}

Unlike the carry-in workload, the carry-out workload generated when all subtasks execute for their full WCETs 
is not guaranteed to be the maximum. 
Consider an interfering task $\tau_i$ shown in Figure~\ref{fig:example_task} and a carry-out window of length 
3 time units. If all subtasks of the carry-out job of $\tau_i$ execute for their WCETs, the carry-out workload 
would be 4 time units, as shown in Figure~\ref{fig:carryout_wcet}. 
However, if subtask $v_{i, 1}$ finishes immediately, i.e., 
executes for 0 time units, the carry-out workload would be 7 time units, as shown in 
Figure~\ref{fig:carryout_nonwcet}. From Lemma~\ref{lem:asap} and the discussion above, to compute an 
upper-bound for carry-out workload we must consider all possible execution times of the subtasks and 
subtasks must execute as soon as they are ready. 

For each subtask $v_{i, a}$ of the carry-out job of an interfering task $\tau_i$, we define two non-negative integer 
variables $X_{i, a}\geq 0 $ and $W_{i, a}\geq 0$. $X_{i, a}$ represents the actual execution time of subtask 
$v_{i, a}$ in the carry-out job and $W_{i, a}$ denotes the contribution of subtask $v_{i, a}$ to 
the carry-out workload. Let $\Delta^{CO}$ be an integer constant denoting the length of the 
carry-out window. Then the carry-out workload is the sum of the contributions of all subtasks in 
$\mathcal{SCHE}^{CO}(\tau_i)$, which is upper-bounded by the maximum of the following 
\emph{optimization objective function}: 
\begin{equation}
\label{eqn:objective}
obj(\tau_i, \Delta^{CO}) \triangleq \sum\limits_{v_{i, a}\in V_i} W_{i, a}.
\end{equation}

The optimal value for the above objective function gives the actual maximum workload generated by the 
carry-out job with unrestricted number of processors. 
We now construct a set of constraints on the contribution of each subtask in 
$\mathcal{SCHE}^{CO}(\tau_i)$ to the carry-out workload. 
From the definitions of $X_{i, a}$ and $W_{i, a}$, we have the following bounds for them. 
\begin{constraint}
\label{con:exetime}
For any interfering task $\tau_i$: \\
$$\forall v_{i, a}\in V_i: 0\leq X_{i, a}\leq C_{i, a}.$$
\end{constraint}

\begin{constraint}
\label{con:workload}
For any interfering task $\tau_i$: \\
$$\forall v_{i, a}\in V_i: 0\leq W_{i, a}\leq X_{i, a}.$$
\end{constraint}

These two constraints come from the fact that the actual execution time of subtask $v_{i, a}$ 
cannot exceed its WCET, and each subtask can contribute at most its whole execution time to 
the carry-out workload. Let $S_{i, a}$ be the starting time of $v_{i, a}$ in $\mathcal{SCHE}^{CO}(\tau_i)$ 
assuming that the carry-out job starts at time instant 0. For simplicity of exposition, we assume 
that the DAG $G_i$ has exactly one source vertex and one sink vertex. If this is not the case, we 
can always add a couple of dummy vertices, $v_{i, source}$ and $v_{i, sink}$, with zero WCETs for 
source and sink vertices, respectively. Then we add edges from $v_{i, source}$ to all vertices with 
no predecessors in the original DAG $G_i$, and edges from all vertices with no successors in $G_i$ 
to $v_{i, sink}$. Without loss of generality, we assume that $v_{i, 1}$ and $v_{i, n_i}$ are the source 
vertex and sink vertex of $G_i$, respectively. 
Let $\sigma^p_{i, a}$ denote a path from the source 
$v_{i, 1}$ to $v_{i, a}$: $\sigma^p_{i, a}\triangleq (v_{i, j_1}, ..., v_{i, j_p})$, where $j_1 = 1$, 
$j_p = a$, and $(v_{i, j_x}, v_{i, j_{x+1}})$ is an edge in $G_i$ $\forall 1\leq x < p$. 
Let $\mathcal{P}(v_{i, a})$ denote the set of all paths from $v_{i, 1}$ to $v_{i, a}$ in $G_i$: 
$\mathcal{P}(v_{i, a})\triangleq \{ \sigma^p_{i, a} \}$. 
$\mathcal{P}(v_{i, a})$ for all subtasks can be constructed by a graph traversal algorithm. 
For instance, a simple modification of depth-first search would accomplish this.

For a particular path $\sigma^p_{i, a}$, the sum of execution times of all subtasks in this path, 
excluding $v_{i, a}$ is called the \emph{distance} to $v_{i, a}$ with respect to this path. 
We let $D^p_{i, a}$ be a variable denoting the distance to $v_{i, a}$ in path $\sigma^p_{i, a}$. 
We impose the following two straightforward constraints on $D^p_{i, a}$ based on its definition. 
\begin{constraint} 
\label{con:dist_smaller}
For any interfering task $\tau_i$: \\
$$\forall v_{i, a}\in V_i, \forall\sigma^p_{i, a}\in \mathcal{P}(v_{i, a}): 
D^p_{i, a} \leq \sum\limits_{v_{i, j_x}\in\{\sigma^p_{i, a}\setminus v_{i, a}\}} X_{i, j_x}.$$ 
\end{constraint}

\begin{constraint}
\label{con:dist_larger}
For any interfering task $\tau_i$: \\
$$\forall v_{i, a}\in V_i, \forall\sigma^p_{i, a}\in \mathcal{P}(v_{i, a}): 
D^p_{i, a} \geq \sum\limits_{v_{i, j_x}\in\{\sigma^p_{i, a}\setminus v_{i, a}\}} X_{i, j_x}.$$
\end{constraint}

In the schedule $\mathcal{SCHE}^{CO}(\tau_i)$, the starting time $S_{i, a}$ of a subtask $v_{i, a}$ 
cannot be smaller than the distance to $v_{i, a}$ in any path $\sigma^p_{i, a}$. We prove this 
as follows. 
\begin{lemma}
\label{lem:starting_time1}
In the schedule $\mathcal{SCHE}^{CO}(\tau_i)$ of any interfering task $\tau_i$:\\
$$\forall v_{i, a}\in V_i, \forall \sigma^p_{i, a}\in \mathcal{P}(v_{i, a}): 
S_{i, a} \geq D^p_{i, a}.$$
\end{lemma}
\begin{proof}
We prove by contradiction. Let $\sigma^{p*}_{i, a}$ be a path so that the starting time 
$S_{i, a}$ is smaller than $D^{p*}_{i, a}$. Subtask $v_{i, a}$ must be ready to start execution, meaning 
all of its predecessors must finish, at time $S_{i, a}$. Since 
$S_{i, a} < D^{p*}_{i, a}$, there must be a subtask $v_{i, j_x}\in\{\sigma^{p*}_{i, a}\setminus v_{i, a}\}$ 
executing (and thus not finished) at time $S_{i, a}$. Then $v_{i, a}$ cannot be ready at time 
$S_{i, a}$ since it depends on $v_{i, j_x}$. This contradicts the assumption that $v_{i, a}$ is ready 
at $S_{i, a}$ and the lemma follows. 
\end{proof}

In fact, in the schedule $\mathcal{SCHE}^{CO}(\tau_i)$ the starting time $S_{i, a}$ of $v_{i, a}$ 
is equal to the longest distance among all paths to it. 
\begin{lemma}
\label{lem:starting_time2}
In the schedule $\mathcal{SCHE}^{CO}(\tau_i)$ of any interfering task $\tau_i$:\\
$$\forall v_{i, a}\in V_i: 
S_{i, a} = \max\limits_{\sigma^p_{i, a}\in\mathcal{P}(v_{i, a})} D^p_{i, a}.$$ 
\end{lemma}
\begin{proof}
Consider a path $\sigma^{p*}_{i, a}$ constructed as follows. First we take a last-completing 
predecessor of $v_{i, a}$, say $v_{i, j_x}$. Since $v_{i, a}$ executes as soon as it is ready, 
it executes immediately after $v_{i, j_x}$ finishes. We recursively trace back through the last-completing 
predecessors in that way until we reach the source vertex $v_{i, 1}$. Path $\sigma^{p*}_{i, a}$ is then 
constructed by chaining the last-completing predecessors together with $v_{i, a}$. 
We note that any subtask $v_{i, j_x}$ in $\sigma^{p*}_{i, a}$ executes as soon as its immediately 
preceding subtask finishes, since no other predecessors of $v_{i, j_x}$ finish later than it does. 
Therefore, $S_{i, a} = D^{p*}_{i, a}$. From Lemma~\ref{lem:starting_time1}, $\sigma^{p*}_{i, a}$ 
must have the longest distance to $v_{i, a}$ among all paths in $\mathcal{P}(v_{i, a})$. 
Thus the lemma follows. 
\end{proof}

Based on Lemmas~\ref{lem:starting_time1} and~\ref{lem:starting_time2}, we have the following constraint 
for the starting time of $v_{i, a}$. 
\begin{constraint}
\label{con:starttime}
For any interfering task $\tau_i$:\\
$$\forall v_{i, a}\in V_i, \forall\sigma^p_{i, a}\in\mathcal{P}(v_{i, a}): 
S_{i, a} \geq D^p_{i, a}.$$
\end{constraint}
\begin{proof}
We prove that this constraint requires that $S_{i, a}$ of every subtask $v_{i, a}$ for which 
$\max_{\sigma^p_{i, a}\in\mathcal{P}(v_{i, a})} D^p_{i, a} < \Delta^{CO}$ 
satisfies Lemma~\ref{lem:starting_time2}, that is 
$S_{i, a} = \max_{\sigma^p_{i, a}\in\mathcal{P}(v_{i, a})} D^p_{i, a}$. 
(Recall that $\Delta^{CO}$ is a constant denoting the carry-out window's length.) 
In other words, we prove that 
it requires that every subtask $v_{i, a}$, which would start executing within the carry-out window in an 
unrestricted-processor schedule $\mathcal{SCHE}^{CO}(\tau_i)$, gets exactly the same starting 
time from the solution to the optimization problem. 
Let $\mathcal{Q}_i$ denote the collection of such subtasks 
--- the ones that would start executing within the carry-out window in $\mathcal{SCHE}^{CO}(\tau_i)$. 

Let $\pi^*$ be the solution to the optimization problem and $S^*_{i, a}$ be the corresponding 
value for the starting time of any subtask $v_{i, a}\in\mathcal{Q}_i$ in the solution $\pi^*$. Obviously 
$S^*_{i, a} \geq \max_{\sigma^p_{i, a}\in\mathcal{P}(v_{i, a})} D^p_{i, a}$ for any $v_{i, a}$ 
since any solution to the optimization problem satisfies this constraint. 
If $S^*_{i, a} = \max_{\sigma^p_{i, a}\in\mathcal{P}(v_{i, a})} D^p_{i, a}$ for any 
$v_{i, a}\in\mathcal{Q}_i$, then we are done. Suppose instead that 
$S^*_{i, a} = \max_{\sigma^p_{i, a}\in\mathcal{P}(v_{i, a})} D^p_{i, a} + \epsilon_{i, a}$, 
$\epsilon_{i, a} > 0$ for some $v_{i, a}\in\mathcal{Q}_i$. Let $\mathcal{Q'}_i$ denote 
the set of such subtasks. 
We construct a solution $\pi'$ to the optimization problem from $\pi^*$ as follows. 
Consider a first subtask $v_{i, a}\in\mathcal{Q'}_i$ in time. We reduce its starting time 
by $\epsilon_{i, a}$: $S'_{i, a} = S^*_{i, a} - \epsilon_{i, a}$. Since $v_{i, a}$ is the first 
delayed subtask, doing this does not violate the precedence constraints for other subtasks. 
We iteratively perform that operation for other subtasks in $\mathcal{Q'}_i$ in increasing time order. 
The solution $\pi'$ constructed in this way yields a larger carry-out workload since more workload 
from individual subtasks can fit in the carry-out window. Therefore $\pi'$ is a better solution, which 
contradicts the assumption that $\pi^*$ is an optimal solution. 
\end{proof}

The workload contributed by a subtask $v_{i, a}$ is: \\
$W_{i, a} = \min\big\{ \max\{ \Delta^{CO} - S_{i, a}, 0 \}, X_{i, a} \big\}$. 
The second part of the outer minimization has been taken care of by Constraint~\ref{con:workload}. 
We now construct constraints to impose the first part of the minimization. 
Let $M_{i, a}$ be an integer variable representing the expression $\max\{ \Delta^{CO} - S_{i, a}, 0 \}$. 
Let $A_{i, a}$ be a binary variable which takes value either 0 or 1. 
We have the following constraints. 
\begin{constraint}
\label{con:workload2}
For any interfering task $\tau_i$:\\
$$ \forall v_{i, a}\in V_i: W_{i, a} \leq M_{i, a}.$$
\end{constraint}

\begin{constraint}
\label{con:max1}
For any interfering task $\tau_i$:\\
$$ \forall v_{i, a}\in V_i: M_{i, a}\geq 0.$$
\end{constraint}

\begin{constraint}
\label{con:max2}
For any interfering task $\tau_i$:\\
$$ \forall v_{i, a}\in V_i: M_{i, a}\leq (\Delta^{CO} - S_{i, a}) A_{i, a}.$$
\end{constraint}

Constraints~\ref{con:max1} and~\ref{con:max2} bound the value for $M_{i, a}$ and 
Constraint~\ref{con:workload2} enforces another upper bound for the workload $W_{i, a}$. 
If $\Delta^{CO} < S_{i, a}$, $A_{i, a}$ can only be 0 in order to satisfy both 
Contraints~\ref{con:max1} and~\ref{con:max2}. If $\Delta^{CO} = S_{i, a}$, the value 
of $A_{i, a}$ does not matter. In both cases, these three constraints together with 
Constraint~\ref{con:workload} bound $W_{i, a}$ to zero contribution of $v_{i, a}$ to the carry-out 
workload. If $\Delta^{CO} > S_{i, a}$, the maximizing process enforces that $A_{i, a}$ takes 
value 1. Therefore in any case Constraints~\ref{con:workload},~\ref{con:workload2},~\ref{con:max1}, 
and~\ref{con:max2} enforce a correct value for the workload contribution $W_{i, a}$ of $v_{i, a}$.

We have constructed an ILP with a quadratic constraint (Constraint~\ref{con:max2}) 
for each $v_{i, a}$, for which the optimal solution value is an upper bound for the carry-out workload. The 
carry-out workload of $\tau_i$ in a carry-out window of length $\Delta^{CO}$ can also 
be upper-bounded by the following straightforward lemma. 
\begin{lemma}
\label{lem:carryout_bound}
The carry-out workload of an interfering task $\tau_i$ scheduled by G-FP in a carry-out 
window of length $\Delta^{CO}$ is upper-bounded by $m\Delta^{CO}$.
\end{lemma}

Lemma~\ref{lem:carryout_bound} follows directly from the fact that the carry-out job 
can execute at most on all $m$ processors of the system during the carry-out window. 
Since the carry-out workload of $\tau_i$ is upper-bounded by both the maximum value 
returned for the optimization problem and Lemma~\ref{lem:carryout_bound}, 
it is upper-bounded by the minimum of the two quantities. 
\begin{theorem}
\label{thm:carryout}
The carry-out workload of an interfering task $\tau_i$ scheduled by G-FP in a carry-out 
window of length $\Delta^{CO}$ is upper-bounded by: 
$\min\Big\{ \mathcal{OBJ}, m\Delta^{CO} \Big\}$, where 
$\mathcal{OBJ}$ is the maximum value returned for the maximization problem (Equation~\ref{eqn:objective}).
\end{theorem}

As discussed in Section~\ref{sec:sota}, the technique proposed by Fonseca et al.~\cite{fonseca2017improved} 
can be applied directly for NFJ-DAGs but not for general DAGs. For a general DAG, the procedure to transform 
the general DAG to an NFJ-DAG will likely inflate the carry-out workload bound as it removes some precedence 
constraints between subtasks and enables a higher parallelism (and thus a greater interfering workload) 
for the carry-out job. In contrast, our method 
directly bounds the carry-out workload for any DAG and the optimal value obtained is the actual 
maximum carry-out workload. Hence, our method theoretically yields better schedulability 
than~\cite{fonseca2017improved}'s for general DAGs. 
The cost of our method is higher time complexity for computing carry-out 
workload due to the hardness of the ILP problem. However, it can be implemented and works effectively 
with modern optimization solvers, as we show in our experiments (Section~\ref{sec:evaluation}).

\section{Response-Time Analysis}
\label{sec:rta_schedtest}

\begin{algorithm}[ht]
\caption{Response-Time Analysis}
\label{algo:sched_test}
\begin{algorithmic}[1]

\Procedure{SchedulabilityTest}{$\tau$}
\Comment{Without loss of generality, assuming tasks are sorted in decreasing order of priority}
\For{Each $\tau_k\in\tau$} \label{alg:ln:initstart}
	\Comment{\parbox[t]{.5\linewidth}{Initialize the values for response-time bounds}}
	\State $R_k^{ub} \leftarrow L_k + \frac{1}{m}(C_k - L_k)$
	\If{Any $R_k^{ub} > D_k$}
		\State Return \emph{Unschedulable}
	\EndIf
\EndFor \label{alg:ln:initend}
\item[]

\For{$\tau_k$ from $\tau_2$ to $\tau_n$}
	\State Calculate $R_k^{ub}$ in Theorem~\ref{thm:responsetime_bound} \label{alg:ln:bound}
	\If{$R_k^{ub} > D_k$} \label{alg:ln:checkstart}
		\State Return \emph{Unschedulable}
	\EndIf \label{alg:ln:checkend}
\EndFor

\State Return \emph{Schedulable} \label{alg:ln:ok}
\EndProcedure

\end{algorithmic}
\end{algorithm}

From the above calculations for the bounds of intra-task interference and inter-task interference 
on $\tau_k$, we have the following theorem for the response-time bound of $\tau_k$.
\begin{theorem}
\label{thm:responsetime_bound}
A constrained-deadline task $\tau_k$ scheduled by a global fixed-priority algorithm has response-time 
upper-bounded by the smallest integer $R_k^{ub}$ that satisfies the following fixed-point iteration: \\
$$R_k^{ub} \leftarrow L_k + \frac{1}{m}(C_k - L_k) + \frac{1}{m}\sum\limits_{\tau_i\in hp(\tau_k)} W_i(R_k^{ub}).$$
\end{theorem}
\begin{proof}
This follows from Equation~\ref{eqn:resptime}, Lemma~\ref{lem:intra_interference} and the fact that the 
inter-task interference of $\tau_i$ on $\tau_k$ is bounded by the workload generated by $\tau_i$ 
(Equation~\ref{eqn:workload_relation}). 
\end{proof}

In Theorem~\ref{thm:responsetime_bound}, $W_i(R_k^{ub})$ is computed using Equation~\ref{eqn:max_workload} 
for all carry-in and carry-out windows that satisfy Equation~\ref{eqn:ci_co_length}. 
For specific carry-in and carry-out window lengths, the carry-in workload is bounded using 
Equation~\ref{eqn:carryin} and the carry-out workload is bounded as discussed in Section~\ref{sec:carryout}. 
The lengths for carry-in window $\Delta_i^{CI}$ and carry-out window 
$\Delta_i^{CO}$ are varied as follows. 
Let $\Gamma$ denote the right-hand side of Equation~\ref{eqn:ci_co_length}. 
First $\Delta_i^{CI}$ takes its largest value: $\Delta_i^{CI}\leftarrow \min\{\Gamma, L_i\}$, and $\Delta_i^{CO}$ takes 
the remaining sum: $\Delta_i^{CO}\leftarrow \min\{\Gamma-\Delta_i^{CI}, L_i\}$. Then in each subsequent step, 
$\Delta_i^{CI}$ is decreased and $\Delta_i^{CO}$ is increased until $\Delta_i^{CO}$ takes its largest value and 
$\Delta_i^{CI}$ takes the remaining value. 
We note that if at the first step both $\Delta_i^{CI}$ and $\Delta_i^{CO}$ are greater than or equal to 
$L_i$, the carry-in workload and carry-out workload are bounded by $\min(C_i, m\Delta_i^{CI})$ and 
 $\min(C_i, m\Delta_i^{CO})$, respectively. 
Similarly, if the sum of $\Delta_i^{CI}$ and $\Delta_i^{CO}$ is 0 in Equation~\ref{eqn:ci_co_length}, both the 
carry-in workload and the carry-out workload are 0. 
We also note that for the highest priority task, there is no interference from any other task, and thus its 
response-time bound can be computed simply by: $R_k^{ub} \leftarrow\big(L_k + \frac{1}{m}(C_k - L_k)\big)$. 

Using the above response-time bound, we derive a schedulability test, shown in Algorithm~\ref{algo:sched_test}. 
First we initialize the response-times for the tasks to be $\big(L_k + \frac{C_k - L_k}{m}\big)$ for all tasks $\tau_k$. 
If for any task, the initial response-time is larger than its relative deadline, then the task set is deemed 
unschedulable (lines~\ref{alg:ln:initstart}-\ref{alg:ln:initend}). 
Otherwise, we repeatedly compute the response-time bound for each task in descending order of 
priority using the fixed-point iteration in Theorem~\ref{thm:responsetime_bound} (line~\ref{alg:ln:bound}). 
After the computation for each task finishes, we check whether the response-time bound is larger than its deadline. 
If it is, then the task set is deemed unschedulable (lines~\ref{alg:ln:checkstart}-\ref{alg:ln:checkend}). 
Otherwise, the task set is deemed schedulable after all tasks have been checked (line~\ref{alg:ln:ok}). 

As expected for response-time analysis, for each task $\tau_i$ the number of iterations in the fixed-point equation 
(Theorem~\ref{thm:responsetime_bound}) is pseudo-polynomial in the task's deadline $D_i$ 
(line~\ref{alg:ln:bound}). In each iteration of the fixed-point equation and for each interfering task, 
we consider all combinations of carry-in and carry-out window lengths that satisfy Equation~\ref{eqn:ci_co_length} 
to compute the maximum interfering workload. There are $\mathcal{O}(L_i)$ such combinations, and thus the 
ILP for the carry-out workload is solved $\mathcal{O}(L_i)$ times. 
The maximum workload over all combinations of carry-in and carry-out window lengths gives an upper-bound 
for the interfering workload generated by the given interfering task.

\section{Evaluation}
\label{sec:evaluation}

\begin{figure*}[t]
\centering
\begin{subfigure}{0.45\textwidth}
  \centering
  \includegraphics[width=\linewidth]{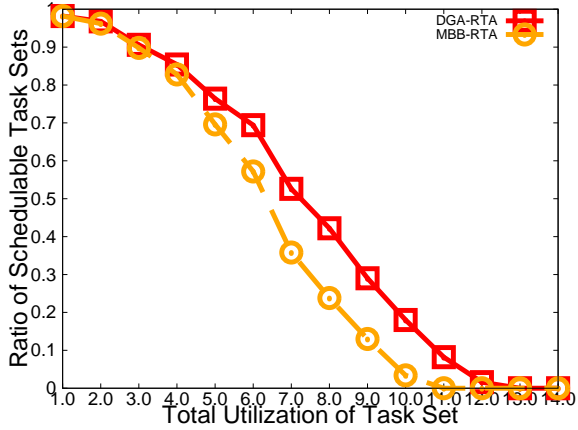}
  \caption{Result for $m=16$, minimum task utilization $\beta = 0.2$, and varying total utilization.}
  \label{fig:varyutil1}
\end{subfigure}
\hfill
\begin{subfigure}{.45\textwidth}
  \centering
  \includegraphics[width=\linewidth]{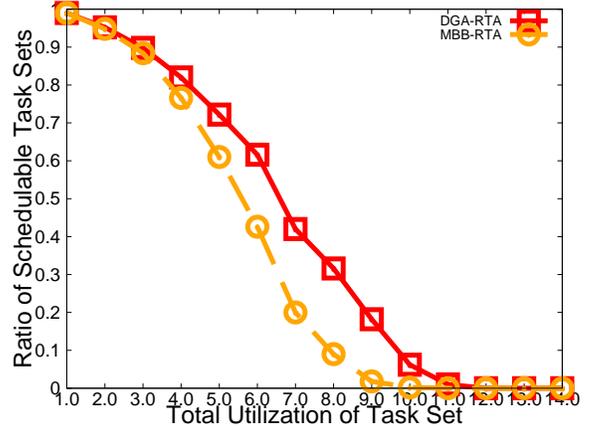}
  \caption{Result for $m=16$, minimum task utilization $\beta = 0.4$, and varying total utilization.}
  \label{fig:varyutil2}
\end{subfigure}

\begin{subfigure}{0.45\textwidth}
  \centering
  \includegraphics[width=\linewidth]{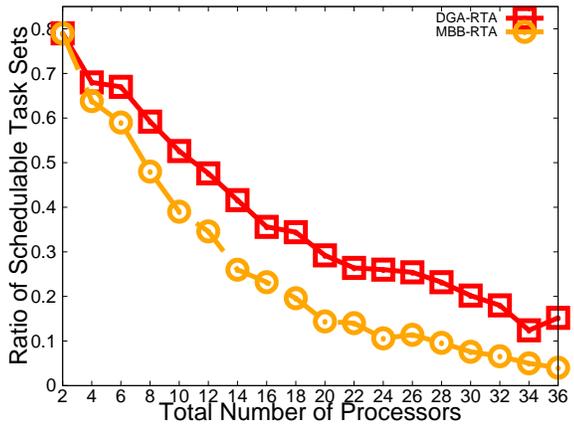}
  \caption{Result for total utilization $U = 0.5m$, minimum utilization $\beta=0.1$, and varying $m$.}
  \label{fig:varyprocs1}
\end{subfigure}
\hfill
\begin{subfigure}{.45\textwidth}
  \centering
  \includegraphics[width=\linewidth]{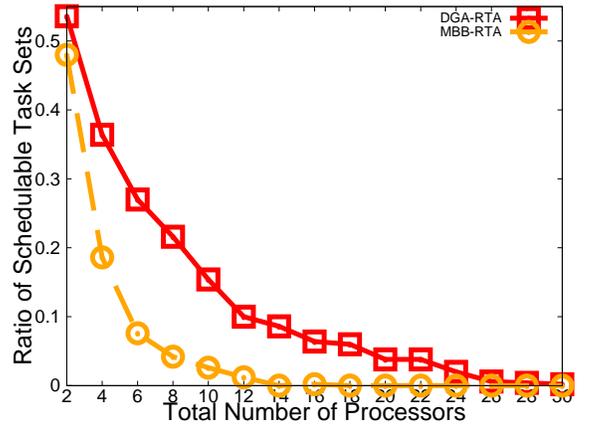}
  \caption{Result for total utilization $U = 0.7m$, minimum utilization $\beta=0.1$, and varying $m$.}
  \label{fig:varyprocs2}
\end{subfigure}

\caption{Ratio of schedulable task sets for varying total utilization and varying number of processors.}
\label{fig:eval}
\end{figure*}

As we discussed in Sections~\ref{sec:sota} and~\ref{sec:carryout}, we apply a similar, high-level framework for 
analyzing schedulability of G-FP scheduling to the one used by Fonseca et al.~\cite{fonseca2017improved} 
--- i.e., accounting for the interfering workloads caused by the body jobs, the carry-in and carry-out jobs 
separately, and maximizing the interference by sliding the problem window. 
However, unlike~\cite{fonseca2017improved} our technique for bounding carry-out workload works 
directly for general DAGs and does not introduce pessimism due to the removal of precedence constraints 
between subtasks, as presented in~\cite{fonseca2017improved}. Though for carry-in workload, we reuse the result 
from~\cite{fonseca2017improved}. Hence, we consider our work as a generalization/extension 
of~\cite{fonseca2017improved} that can be applied for general sporadic DAG tasks. 
The performance of our method in term of schedulability ratio is compatible with~\cite{fonseca2017improved}'s 
--- it theoretically is at least as good as~\cite{fonseca2017improved} for NFJ-DAGs and is better 
than~\cite{fonseca2017improved} for non NFJ-DAGs. We thus focus on measuring the performance of 
our method and use the work by Melani et al.~\cite{melani2015response} as a reference for evaluating 
the improvement of our method upon their simple one. 

We applied the Erd\H{o}s-R\'enyi $G(n, p)$ method, described in~\cite{cordeiro2010random}, to generate 
DAG tasks. In this method the number of subtasks, given by parameter $n$ in $G(n, p)$, is first fixed. 
Then, directed edges between pairs of vertices are added with probability $p$. Since the obtained DAG may not 
necessarily be connected, we added a minimum number of edges to make it weakly connected. 
In our experiments, the probability for a directed edge to be added is $p=0.2$. 
We chose the number of subtasks uniformly in the range $[10, 20]$. 
Other parameters for each DAG task $\tau_i$ were generated similarly to~\cite{melani2015response}. 
In particular, the WCETs of subtasks of $\tau_i$ were generated uniformly in the range $[1, 100]$. 
After that, the work $C_i$ and span $L_i$ were calculated. $\tau_i$'s utilization was generated uniformly 
in the range $[\beta, C_i/L_i]$, where $\beta\leq 1$ is a parameter to control the minimum task's utilization and 
$C_i/L_i$ represents the degree of parallelism of task $\tau_i$. $\tau_i$'s deadline $D_i$ was 
generated using a normal distribution with mean equal to $(\frac{T_i+L_i}{2})$ and 
standard deviation equal to $(\frac{T_i-L_i}{4})$. We kept generating the relative deadline until a value in 
the range $[L_i, T_i]$ was obtained. 

To generate a task set for a given total utilization, we repeatedly add DAG tasks to the task set until the 
desired utilization is reached. The utilization (and period) of the last task may need to be adjusted to 
match the total utilization. We used the SCIP solver~\cite{scip} with CPLEX~\cite{cplex} as 
its underlying LP-solver to compute the bound for carry-out workload. 
For our experiments, we set the default minimum utilization of individual tasks $\beta$ to $0.1$. 
For each configuration we generated 
500 task sets and recorded the ratios of task sets that were deemed schedulable. We compare our response-time 
analysis, denoted by DGA-RTA, with the response-time analysis introduced in~\cite{melani2015response}, 
denoted by MBB-RTA. 
For all generated task sets, priorities were assigned in Deadline Monotonic order 
--- studying an efficient priority assignment scheme for G-FP is beyond the scope of this paper. 

Figures~\ref{fig:varyutil1},~\ref{fig:varyutil2},~\ref{fig:varyprocs1}, and~\ref{fig:varyprocs2} show representative 
results for our experiments. In Figure~\ref{fig:varyutil1} and~\ref{fig:varyutil2}, we fixed the total number 
of processors $m=16$ and varied the total utilization from 1.0 to 14.0. The minimum task utilization $\beta$ was 
set to $0.2$ and $0.4$ in these two experiments, respectively. Unsurprisingly, DGA-RTA dominates MBB-RTA, 
as also observed in~\cite{fonseca2017improved}. Notably, its schedulability ratios for some configurations are 
two times or more greater than MBB-RTA, e.g., for total utilizations of 8.0, 9.0 in Figure~\ref{fig:varyutil1}, and 
7.0, 8.0 in Figure~\ref{fig:varyutil2}. 
In Figures~\ref{fig:varyprocs1} and~\ref{fig:varyprocs2}, we fixed the normalized total utilization and 
varied the number of processors $m$ from 2 to 36. 
For each value of $m$, we generated task sets with total utilization $U = 0.5m$ or $U = 0.7m$ for these two 
experiments, respectively. Similar to the previous experiments, the schedulability ratios of the generated 
task sets were improved significantly using DGA-RTA compared to MBB-RTA. 

To provide a trade-off between computational complexity and accuracy of schedulability test, 
one can employ our analysis in combination with the analysis 
presented in~\cite{fonseca2017improved} by first applying their response-time analysis and then using our 
analysis if the task set is deemed unschedulable by~\cite{fonseca2017improved}. In this way, one can get 
the best result from both analyses.

\section{Conclusion}
\label{sec:conclusion}

In this paper we consider constrained-deadline, parallel DAG tasks scheduled under 
a preemptive, G-FP scheduling algorithm on multiprocessor platforms. 
We propose a new technique for bounding carry-out workload 
of interfering task by converting the calculation of the bound to an optimization problem, 
for which efficient solvers exist. The proposed technique applies directly to general DAG tasks. 
The optimal solution value for the 
optimization problem serves as a safe and tight upper bound for carry-out workload. 
We present a response-time analysis for G-FP based on the proposed workload bounding technique. 
Experimental results affirm the dominance of the proposed approach over existing techniques. 
There are a couple of open questions that we would like to address in future. They include 
bounding carry-in and carry-out workloads for the actual number of processors $m$ of the system and 
designing an efficient priority assignment scheme for parallel DAG tasks scheduled under G-FP algorithm.

\bibliographystyle{ACM-Reference-Format}
\bibliography{ms}

\end{document}